\documentclass[ aps,
                pra,
                showpacs,
                amssymb,
                nofootinbib,
                superscriptaddress,
                twocolumn,
                longbibliography]{revtex4-1}
\usepackage[ansinew]{inputenc}
\usepackage{bbm}
\usepackage{bm}
\usepackage{amsbsy}
\usepackage{amsthm}
\usepackage{amssymb}
\usepackage{amsfonts}
\usepackage{amsmath, mathtools}
\usepackage{dsfont}
\usepackage{graphicx}
\usepackage{epsfig}
\usepackage{epstopdf}
\usepackage{dsfont}
\usepackage{multibib}
\usepackage{enumitem}
\usepackage{xcolor}

\usepackage[colorlinks]{hyperref}
\makeatletter
\newcommand\org@hypertarget{}
\let\org@hypertarget\hypertarget
\renewcommand\hypertarget[2]{%
  \Hy@raisedlink{\org@hypertarget{#1}{}}#2%
  }
\makeatother
\usepackage[figure,table]{hypcap}
\usepackage{MnSymbol}
\usepackage{float}
\usepackage{comment}

\usepackage{subfigure}

\hypersetup{
	bookmarksnumbered,
	pdfstartview={FitH},
	citecolor={darkgreen},
	linkcolor={darkred},
	urlcolor={darkblue},
	pdfpagemode={UseOutlines}}
\definecolor{darkgreen}{RGB}{50,190,50}
\definecolor{darkblue}{RGB}{0,0,190}
\definecolor{darkred}{RGB}{238,0,0}
\definecolor{darkpurple}{rgb}{0.4, 0.2, 0.5}
\usepackage{soul}

\newcommand{\ket}[1]{\ensuremath{\left|\right.\!{#1}\!\left.\right\rangle}}

\newcommand{\id}{\mathds{1}}

\DeclareMathOperator{\Tr}{Tr}

\newcommand{\djj}{d\kern-0.4em\char"16\kern-0.1em}
\newtheorem{theorem}{Theorem}

\newtheorem*{lemma*}{Lemma}

\makeatletter
\renewcommand{\p@subsection}{}
\renewcommand{\p@subsubsection}{}
\makeatother

\begin{document}
\title{Correlation constraints and the Bloch geometry of two qubits}

\author{Simon Morelli}
\email{smorelli@bcamath.org}
\affiliation{BCAM - Basque Center for Applied Mathematics,
Mazarredo 14, E48009 Bilbao, Basque Country - Spain}
\author{Christopher Eltschka}
\affiliation{Institut f\"ur Theoretische Physik, Universit\"at Regensburg, D-93040 Regensburg, Germany}
\author{Marcus Huber}
\affiliation{Atominstitut,  Technische  Universit{\"a}t  Wien,  1020  Vienna,  Austria}
\affiliation{Institute for Quantum Optics and Quantum Information - IQOQI Vienna, Austrian Academy of Sciences, Boltzmanngasse 3, 1090 Vienna, Austria}

\author{Jens Siewert}
\email{jens.siewert@ehu.eus}
\affiliation{
Department of Physical Chemistry and EHU Quantum Center, University of the Basque Country UPV/EHU,
E-48080 Bilbao, Spain}
\affiliation{IKERBASQUE Basque Foundation for Science, E-48013 Bilbao, Spain}

\date{\today}

\begin{abstract}
    We present a novel inequality on the purity of a bipartite state depending solely on the length-difference of the local Bloch vectors. For two qubits this inequality is tight for all marginal states and so extends the previously known solution for the 2-qubit marginal problem.
    With this inequality we construct a 3-dimensional Bloch model of the 2-qubit quantum state space in terms of Bloch lengths, providing a pleasing visualization of this high-dimensional state space.
    This allows to characterize quantum states by a strongly reduced set of parameters and to investigate the interplay between local properties of the marginal systems and global properties encoded in the correlations.
    
\end{abstract}
\maketitle


\section{Introduction}
The investigation of fundamental bounds in the form of entropy inequalities, limiting the distribution of information within a multipartite system, marked the birth of information theory as a field of study~\cite{Shannon48}.
In the same way, the exploration of fundamental limitations on the distribution of quantum information within a system lies at the core of the rapidly evolving field of quantum information theory.
Analogous results to classical information theory were found for the von~Neumann entropy~\cite{Araki70,LiebRusskai73,Pippenger03}, but at the same time there remain many open questions~\cite{LindenWinter05,CadneyLindenWinter12}.

Moreover, these are not the only known constraints to the distribution of quantum information within a system. Monogamy relations bound the shared entanglement that a party can have with other parties and thus ultimately limit the set of attainable quantum states~\cite{Coffman_2000,OsborneVerstraete2006,Eltschka2009,Bai2014,Lancien2016,Adesso2016,Meyer2017,Camalet2017,Gour2018,EltschkaSiewert2018,EltschkaHuberGuhneSiewert18}.

The classical probability space is described geometrically by a simplex of the corresponding dimension.
While the latter geometrically is fairly regular, the quantum state space becomes a highly complex object that admits no simple description~\cite{Bengtsson_2006,Bengtsson_2012}.
Ultimately the characterization of attainable quantum states is linked to the famous quantum marginal problem, investigating the compatibility between local and global eigenvalues of a state~\cite{Klyachko2006,Yu_2021}.

In this work we present a novel constraint of the purity of any bipartite quantum state, based solely on the difference of the local Bloch lengths, i.e., the Hilbert-Schmidt distance of the marginal states to the normalized identity.
For two qubits this constraint is tight for all possible marginal states, thus extending the previously known solution for the marginal problem for two qubits~\cite{Bravyi2003}. 
The bound can be interpreted as a new constraint on the linear entropy of a bipartite system, as the linear entropy has a one-to one correspondence to the Bloch lengths, and also be extended to a constraint of the local Bloch lengths of pure tripartite states.

We then use this novel inequality to construct a three dimensional model for the state space of two qubits, where we use the two local Bloch lengths and the correlation tensor length as the coordinates.
This allows to characterize important properties of states based on a very reduced set of parameters.
In the model we further identify regions of special interest, such as purely entangled and purely separable regions. It is worthwhile mentioning here that all our discussion and results are based on the lengths of Bloch vectors, which are directly accessible in experiments and, because of their local unitary invariance, are particularly interesting in the context of the arising field of randomized measurements~\cite{Tran2016,Knips2020,Elben2022}.


\section{Quantum state space and the Bloch representation}

In what is known as the Bloch representation, density operators of qubits are parametrized by the Pauli matrices $\sigma_x$, $\sigma_y$, $\sigma_z$ and the identity $\id_2$
\begin{align}
     \rho\ =\ \frac{1}{2}(\id_2+x\ \sigma_x + y\ \sigma_y +z\ \sigma_z ),
\label{eq:bloch2}
\end{align}
where $x$, $y$ and $z$ are real numbers satisfying $x^2+y^2+z^2\le 1$~\cite{Fano1954,Fano1957}.
Hence, the state space $\mathcal{Q}_2$ of a qubit can be identified with a three-dimensional ball of radius 1 about the origin, with a one to one correspondence between Bloch vectors $\mathbf{v}=(x,y,z)$ and states $\rho$.

Unfortunately, the Bloch ball falls short of giving an extensive comprehension of the state space of quantum systems, not because of shortcomings of the geometric model but because the system of a single qubit itself is too elementary to exhibit all the complex features of a general quantum system.
Looking at higher-dimensional systems gives a richer and more diverse structure of the quantum state-space and its geometrical representation, e.g., Refs.~\cite{Bengtsson_2006,Bloore_1976,Kimura_2003,Kurzynski_2009,Sarbicki_2012,Bengtsson_2012,Goyal_2016,Sharma_2021,rau2021symmetries,Eltschka_2021}.
The Bloch representation exists also for higher-dimensional systems. We call a basis $\{\mu_i\}_{i=0}^{d^2-1}$ satisfying $\mu_0=\id_d$ and 
$\Tr(\mu_i\mu_j^{\dagger})=d\delta_{ij}$ a Bloch basis~\cite{BertlmannKrammer08}. Every state of a $d$-dimensional system can be expanded in such a basis as
\begin{align}
     \rho\ =\ \frac{1}{d}(\id_d+\sum\limits_{i=1}^{d^2-1} v_i \mu_i),
\label{eq:blochd}
\end{align}
where we do not include $v_0=1$ in the definition of the Bloch vector $\mathbf{v}$.
While the Bloch vector itself clearly depends on the Bloch basis, the length of the Bloch vector $\|\mathbf{v}\|_2=\sqrt{d\Tr(\rho^2)-1}$ is independent of the choice of Bloch basis as well as state basis, as this quantity is unitarily invariant. To reflect the fact that no basis choice with corresponding Bloch vector is needed to define this quantity, we will refer to it simply as Bloch length.
From now on, if not stated otherwise, all norms of vectors will be the 2-norm and we omit the subscript.

As interesting and complex as these single systems are, they still lack an essential property: they are unable to display correlations~\cite{footnote1} between different parties of a more complex quantum system. Steps in that direction were taken in Ref.~\cite{Gamel_2016,Wie_2020,zyczkowski2022quantum}.
One can extend local Bloch bases to a product basis of a multipartite system, retaining the party structure of the system.
Let $\{\mu_i^A\}$ and $\{\mu_j^B\}$ be Bloch bases for the systems $A$ and $B$ of equal dimension $d$. Then $\{\mu_i^A\otimes\mu_j^B\}$ is a Bloch basis for the joint system and every state can be expressed as
\begin{equation}
	\rho_{AB} = \frac{1}{d^2} \big(\mathds{1}_{d^2} + \sum_{i} a_{i} \mu_i^A\otimes\mathds{1}_d + \sum_{j} b_{j} \mathds{1}_d\otimes\mu_j^B + \sum_{i,j} t_{ij} \mu_i^A\otimes\mu_j^B \big).
\end{equation}
The local Bloch vectors $\mathbf{a}$ and $\mathbf{b}$ encode all the information of the local states $\rho_A=\Tr_B(\rho_{AB})$ and $\rho_B=\Tr_A(\rho_{AB})$ respectively. Correlations~\cite{footnote1} are encoded in the correlation tensor $\mathbf{T}_{AB}$.
These concepts can be straightforwardly generalized for multipartite systems. Here and in the following we consider only systems of equal local dimensions, but all concepts are also well defined for systems of arbitrary local dimensions.


\section{Purity constraint from sector-difference}
An immediate consequence of the Schmidt decomposition is that for every pure bipartite quantum state the marginal states are equal up to local unitaries.
This means that if the eigenvalues of the marginals are equal, there exists no non-trivial constraint for the purity of the global state.
However, if the marginals have different eigenvalues no pure state is compatible with those constraints. Therefore there exists a purity constraint based on the difference of the local states.
The following theorem gives an upper bound on the purity of a bipartite quantum state based solely on the difference of the local Bloch lengths. For qubits this bound is tight.

\begin{theorem}\label{theorem_purity}
For any bipartite state $\rho_{AB}$ with equal local dimension $d$ it holds that
    \begin{align}
        d\Tr(\rho^2_{AB})\le d-\sqrt{2d}\big|\|\mathbf{a}\|-\|\mathbf{b}\|\big|+\big|\|\mathbf{a}\|-\|\mathbf{b}\|\big|^2,\label{eq:trace_bound}
    \end{align}
    where $\mathbf{a}$ and $\mathbf{b}$ are the local Bloch vectors of $\rho_A=\Tr_B(\rho_{AB})$ and $\rho_B=\Tr_A(\rho_{AB})$, respectively,
    and $\|\cdot\|$ denotes the 2-norm.
\end{theorem}

\begin{proof}
Denote by $\Delta=\big|\|\mathbf{a}\|-\|\mathbf{b}\|\big|$.
Note that any state with eigenvalues $\lambda=1-\Delta/\sqrt{2d}$, $1-\lambda=\Delta/\sqrt{2d}$ and $\lambda_j=0$ ($j>2$) satisfies the inequality.
Assume now that there exist a state $\rho$ such that $d\Tr(\rho^2)>d-\sqrt{2d}\Delta+\Delta^2$.
Its purity $\Tr(\rho^2)\le\lambda^2+(1-\lambda)^2$ can only exceed the bound of Eq.~(\ref{eq:trace_bound}) if its leading eigenvalue satisfies $\lambda>1-\Delta/\sqrt{2d}$.
Therefore it must hold that $1-\lambda<\Delta/\sqrt{2d}$.
Any state can be written as
\begin{align}
    \rho=\lambda|\psi\rangle\langle\psi|+(1-\lambda)\sigma,
\end{align}
where $|\psi\rangle$ is the eigenvector of the eigenvalue $\lambda$ and $\sigma$ the sum of the projectors of the remaining eigenvectors in the orthogonal subspace.
Since all quantities in Eq.~\eqref{eq:trace_bound} are invariant under local unitary operations, we can choose any local basis to prove the statement.
Writing this state in the Schmidt basis of $|\psi\rangle$, we find
\begin{align}
    \rho_A&=\lambda\rho_\psi+(1-\lambda)\sigma_A\\
    \rho_B&=\lambda\rho_\psi+(1-\lambda)\sigma_B,
\end{align}
with equal marginals $\rho_\psi=\Tr_B(|\psi\rangle\langle\psi|)=\Tr_A(|\psi\rangle\langle\psi|)$.
Now it follows
\begin{align}
    1-\lambda&<\frac{1}{\sqrt{2d}}\big|\|\mathbf{a}\|-\|\mathbf{b}\|\big|\\
    &=\tfrac{1}{\sqrt{2}d}\big|\|d\rho_A-\mathds{1}\|_{\mathrm{HS}}-\|d\rho_B-\mathds{1}\|_{\mathrm{HS}}\big|\\
    &\le\tfrac{1}{\sqrt{2}}\|\rho_A-\rho_B\|_{\mathrm{HS}}\\
    &=\tfrac{1-\lambda}{\sqrt{2}}\|\sigma_A-\sigma_B\|_{\mathrm{HS}}\\
    &\le 1-\lambda,
\end{align}
where $\|M\|_{\mathrm{HS}}=\sqrt{\Tr(M^\dagger M)}$ denotes the Hilbert-Schmidt or Frobenius norm.
For the last inequality we have used $\|\sigma_A-\sigma_B\|_{\mathrm{HS}}\le\sqrt{2}$, which follows immediately from the positivity of the matrices.
But this is a contradiction, therefore the assumption that there exists a state $\rho$ violating the inequality is wrong.
\end{proof}

\begin{figure}[t]
  \centering
  \includegraphics[width=1\linewidth]{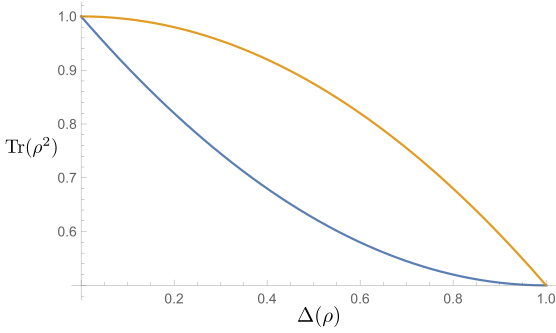}
  \caption{The maximal value of $\Tr(\rho^2)$ plotted against the difference $\Delta(\rho)= \big|\|\mathbf{a}\|-\|\mathbf{b}\|\big|$ for d=2. The bound of Eq.~(\ref{eq:trace_bound}) is shown in blue and the bound following from the triangle inequality of the linear entropy in orange.
    }
  \label{fig:trace_bound}
\end{figure}

This bound improves the previously known bound following from the triangle inequality of the linear entropy, see Fig.~\ref{fig:trace_bound}.
Remarkable is the fact that for qubits this inequality is not only tight with respect to the difference $\big|\|\mathbf{a}\|-\|\mathbf{b}\|\big|$, but actually for all values of $\|\mathbf{a}\|$ and $\|\mathbf{b}\|$, see Eq.~(\ref{eq:boundary}) for a family of states that saturate the inequality.
But the local Bloch length describes the state of the marginal qubits uniquely up to local unitaries, which clearly do not change any of the quantities in the relation. This means, that for every two marginal states, there exists a global state compatible with them that saturates the inequality.
So we give at the same time a more extensive answer to the pure two-qubit marginal problem and a physically meaningful interpretation of the solution of the general two-qubit marginal problem~\cite{Bravyi2003}. We are able to exactly quantify the maximum purity of any global state compatible with any two given marginals (or in fact only the difference of their local Bloch lengths).

A direct consequence of this inequality is that for any pure tripartite state of equal local dimension it holds that the three local Bloch lengths satisfy
\begin{align}\label{eq:pure_qudit_constraint}
    \|\mathbf{c}\|^2\le d-1-\sqrt{2d}\big|\|\mathbf{a}\|-\|\mathbf{b}\|\big|+\big|\|\mathbf{a}\|-\|\mathbf{b}\|\big|^2.
\end{align}
This can be seen by noting that for a pure three-party state it holds that $\Tr(\rho_{AB}^2)=\Tr(\rho_{C}^2)$ and inserting this into Eq.~(\ref{eq:trace_bound}).

For qubits this relation can be simplified to the known inequality
\begin{align}
    \|\mathbf{a}\|+\|\mathbf{c}\|\le1+\|\mathbf{b}\|,
\end{align}
following directly from the solution of the $n$-qubit marginal problem~\cite{Higuchi_2003,Cheng_2017}.
It is saturated by states of the form
\begin{align}
    \ket{\psi}=\sqrt{\tfrac{x-y}{2}}\ket{001}+\sqrt{\tfrac{1+y}{2}}\ket{010}+\sqrt{\tfrac{1-x}{2}}\ket{100}.
\end{align}

While conceptually different, this constraint can be compared to the famous Coffman-Kundu-Wootters  relation~(CKW)~\cite{Coffman_2000}. The states that saturate this relation for three qubits form a strict subset of the states saturating CKW. This means that CKW is sharp in a larger region than this novel constraint, thus proving it to be weaker than CKW. At the same time Eq.~(\ref{eq:pure_qudit_constraint}) holds in all finite dimensions, opening an exciting new avenue for the study of high-dimensional multipartite quantum systems.


\section{A three-dimensional model for two qubits}

By taking tensor products of the Pauli matrices, this matrix-basis for the density operators of a single qubit can be extended to a product basis of the composite space of multiple qubits.
That is, every two-qubit state can be expanded in the product Pauli basis 
\begin{equation}
	\rho_{AB} = \frac{1}{4} \big(\mathds{1}_4 + \sum_{i} a_{i} \sigma_i\otimes\mathds{1}_2 + \sum_{j} b_{j} \mathds{1}_2\otimes\sigma_j + \sum_{i,j} t_{ij} \sigma_i\otimes\sigma_j \big),
\end{equation}
where $i,j=x,y,z$ and $a_{i} = \langle\sigma_i\otimes\mathds{1}_2\rangle\equiv\Tr\left(\rho\ \sigma_i\otimes\mathds{1}_2\right)$, $b_{j} = \langle\mathds{1}_2\otimes\sigma_j\rangle$ and $t_{ij}=\langle \sigma_i\otimes\sigma_j \rangle$.
The state space of a two-qubit system is now characterized by 15 real parameters and can thus be identified with a region in $\mathbb{R}^{15}$.
While the condition $\sum_{i} a_i^2+\sum_{j} b_j^2+\sum_{i,j} t_{ij}^2\le 3$ is still necessary to describe a state, it is not sufficient any longer, see Ref.~\cite{Englert_2001,Gamel_2016} for necessary and sufficient criteria.
While points of this object are bijectively mapped to quantum states, the high dimensionality makes it sometimes difficult to work with and limits its usefulness for visualizations.
Moreover, a full description of the quantum state is often not desired, rather than general constraints on physically relevant quantities, such as purities and entropies.
In the following we focus on a more restricted set of parameters to construct a lower-dimensional model of the quantum state space of a two-qubit system.
We compress the local information of each subsystem into a single coordinate. The state of a qubit is characterized uniquely up to local unitaries by its purity, or alternatively, the length of its Bloch vector.
With this choice we obtain our first two coordinates
\begin{align}
    x=\|\mathbf{a}\|:=\sqrt{\sum_{i}a_i^2},\quad y=\|\mathbf{b}\|:=\sqrt{\sum_{j}b_j^2},
\end{align}
as the local Bloch lengths, i.e., the length of the Bloch vector of the reduced state in the corresponding subsystem.
Finally, the last coordinate describes the strength of (classical and quantum) correlation between the local states.
As the third coordinate we choose 
\begin{align}
    z=\|\mathbf{T}_{AB}\|:=\sqrt{\sum_{i,j}t_{ij}^2},
\end{align}
the remaining length of the global Bloch vector not encoded in the local Bloch vectors, which we call the correlation tensor length $\|\mathbf{T}_{AB}\|$.

With these reduced coordinates we can describe the set of points $(\|\mathbf{a}\|,\|\mathbf{b}\|,\|\mathbf{T}_{AB}\|)$ attainable by quantum states.
Besides positivity of the entries, the set is completely characterized by the inequalities
\begin{align}
    \|\mathbf{T}_{AB}\|\ge&\ \|\mathbf{a}\|+\|\mathbf{b}\|-1\label{eq:lowerbound2}\\
    \|\mathbf{T}_{AB}\|^2\le&\ 3+\|\mathbf{a}\|^2+\|\mathbf{b}\|^2-4\|\mathbf{a}\|\|\mathbf{b}\|-4\big|\|\mathbf{a}\|-\|\mathbf{b}\|\big|,\label{eq:upperbound2}
\end{align}
in the sense that all quantum states are mapped into the set defined by the inequalities and for every point satisfying both inequalities there exists at least one state that is mapped to it.

\begin{figure}[t]
  \centering
  \includegraphics[width=1\linewidth]{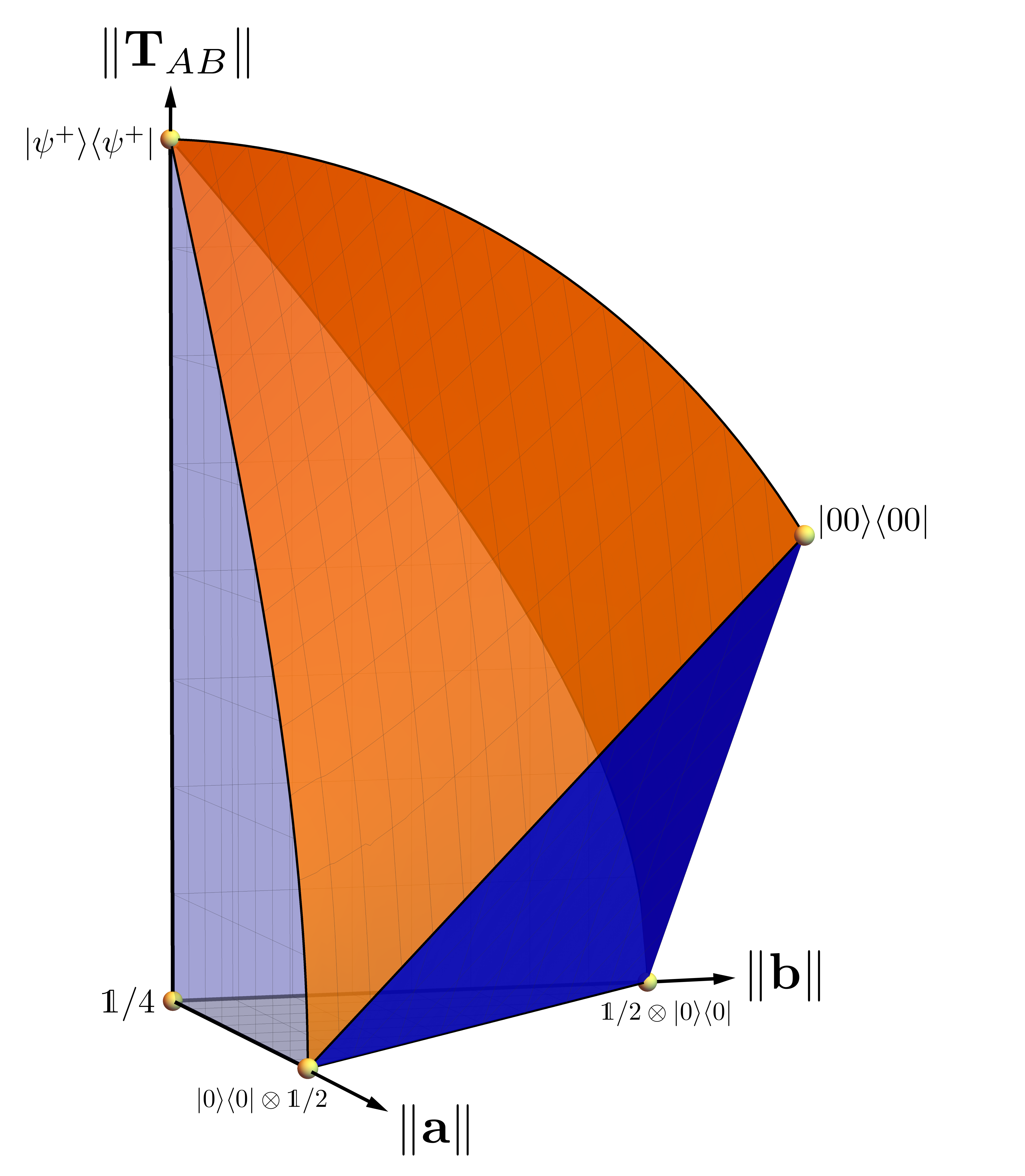}
  \caption{The model of the  Bloch body $\mathcal{Q}_{2\times2}$, depending on the local Bloch lengths $\|\mathbf{a}\|$, $\|\mathbf{b}\|$ and the correlation tensor length $\|\mathbf{T}_{AB}\|$. It is restricted from below by Eq.~(\ref{eq:lowerbound2}) in blue and from above by Eq.~(\ref{eq:upperbound2}) in orange.
    }
  \label{fig:bipartite_body}
\end{figure}

Eq.~(\ref{eq:lowerbound2}) was already proven in~Ref.~\cite{Morelli20} and Eq.~(\ref{eq:upperbound2}) follows directly from Eq.~(\ref{eq:trace_bound}) using the identity $d^2\Tr(\rho_{AB}^2)=1+\|\mathbf{a}\|^2+\|\mathbf{b}\|^2+\|\mathbf{T}_{AB}\|^2$.
Eq.~(\ref{eq:lowerbound2}) cuts out the lower corner of the Bloch model, depicted in blue in Fig.~\ref{fig:bipartite_body}. It is saturated by states of the form
\begin{equation}
    \rho_{\text{lb}}(p,q)=p\ |0\rangle\langle0|\otimes\mathds{1}/2+q\ \mathds{1}/2\otimes|0\rangle\langle0|+(1-p-q)|00\rangle\langle00|
\end{equation}
see Ref.~\cite{Morelli20} for more details.
Eq.~(\ref{eq:upperbound2}) gives an upper bound, shown in orange in Fig.~\ref{fig:bipartite_body}. It is saturated by states of the form
\begin{align}\label{eq:boundary}
    \rho_{\text{ub}}(p,q)=p|\phi(q)\rangle\langle\phi(q)|+(1-p)|ij\rangle\langle ij|,
\end{align}

where $|\phi(q)\rangle=\sqrt{q}|00\rangle+\sqrt{1-q}|11\rangle$ and $i\neq j\in\{0,1\}$.
It follows that $\|\mathbf{a}\|=1-2pj-2pq$, $\|\mathbf{b}\|=1-2pi-2pq$ and $\big|\|\mathbf{a}\|-\|\mathbf{b}\|\big|=2p$. Since $\Tr(\rho_{\text{ub}}(p,q)^2)=1-2p+2p^2$, Eq.~(\ref{eq:trace_bound}), as well as Eq.~(\ref{eq:upperbound2}), are saturated.
The coordinates are length preserving, the length of the global Bloch vector is $\sqrt{\|\mathbf{a}\|^2+\|\mathbf{b}\|^2+\|\mathbf{T}_{AB}\|^2}$, so the distance from the origin in our model remains the Bloch length.
The resulting shape is convex with flat and curved surfaces. 
Since for pure states $\|\mathbf{a}\|=\|\mathbf{b}\|$, all pure states lie on a circular arc of radius $\sqrt{3}$ spanning from $(0,0,\sqrt{3})$, representing maximally entangled states, to $(1,1,1)$, representing product states. This forms the ridge
of the model shown in Fig.~\ref{fig:bipartite_body}.
The maximally mixed state $\mathds{1}_4/4$ is mapped to the origin.
The model has three "artificial" boundaries, since the norms we consider are always non-negative. In this sense, the boundaries where one coordinate is zero are not considered surfaces of the model.

The upper boundary consists of rank-$2$ states, they turn out to be exactly the states maximizing the entanglement for given mixedness of the subsystems discussed later.
The lower boundary consists of states of rank $3$. So all boundaries are occupied only by states of reduced rank. The converse does not hold, i.e., there are states of  rank $2$ and rank $3$ that are mapped to the interior. However, all the states that map to the interior of the ball with radius $1/3$ around the origin are of full rank.

The classical state space of a system with $d$ outcomes is the set of all probability vectors of size $d$ and is described by a simplex $\mathbf{S}_{d-1}$. This set can be identified with the set of diagonal quantum states of the same size.
For a 2-bit classical system, the states, corresponding to diagonal quantum states, form a 3-dimensional subset in the quantum states space in form of the simplex $\mathbf{S}_{3}$, or tetrahedron.
Using a parametrization in the Pauli matrices $(\langle\sigma_z\otimes\mathds{1}_2\rangle,\langle\mathds{1}_2\otimes\sigma_z\rangle,\langle\sigma_z\otimes\sigma_z\rangle)$, this tetrahedron corresponds to the convex hull of the points $(1,1,1)$, $(-1,-1,1)$, $(-1,1,-1)$ and $(1,-1,-1)$.
The intersection with the positive octant is inscribed in our model, described by the convex hull of the points $(1,1,1)$, $(1,0,0)$, $(0,1,0)$, $(0,0,1)$ and $(0,0,0)$.
Looking closely at Eq.~(\ref{eq:trace_bound}) one notices a remarkable fact:
The inequality is saturated everywhere and depends only on the difference of the local Bloch lengths.
But since the trace is invariant under unitary transformations, this means that all boundary states with equal local Bloch length difference have the same distance from the origin.
In fact we notice that the states on the lower edge $\rho_{\text{lb}}(p,0)=p\ |00\rangle\langle00|+(1-p)\ |0\rangle\langle0|\otimes\mathds{1}/2$ form exactly the states on the boundary, if rotated by the entangling unitary 

\begin{align}\label{eq:unitary}
    U(\theta)=&\cos{\theta}(|00\rangle\langle00|+|11\rangle\langle11|)+\sin{\theta}(|11\rangle\langle00|-|00\rangle\langle11|)\notag\\
    &+|01\rangle\langle01|+|10\rangle\langle10|.
\end{align}

Therefore our model is obtained by rotating the classical probability space about the axis $(1,-1,0)$, visualizing the fact that all quantum states can be obtained by an appropriate unitary rotation of a positive and normalized diagonal matrix.


\section{Entangled and separable regions}
One interesting question is to locate entangled states in this picture. The $z$-axis shows the correlation between the two subsystems. This correlation is the sum of the classical and the quantum correlations, see~Ref.~\cite{adesso_arxiv_2016}. Although states mapping to the same point in the model have the same amount of correlation, the respective parts of classical and quantum correlation might differ considerably. Therefore separable and entangled states will not map into two disjoint regions. Nevertheless there can be regions where only entangled or only separable states map into.

\begin{figure}[t]
  \centering
  \includegraphics[width=1\linewidth]{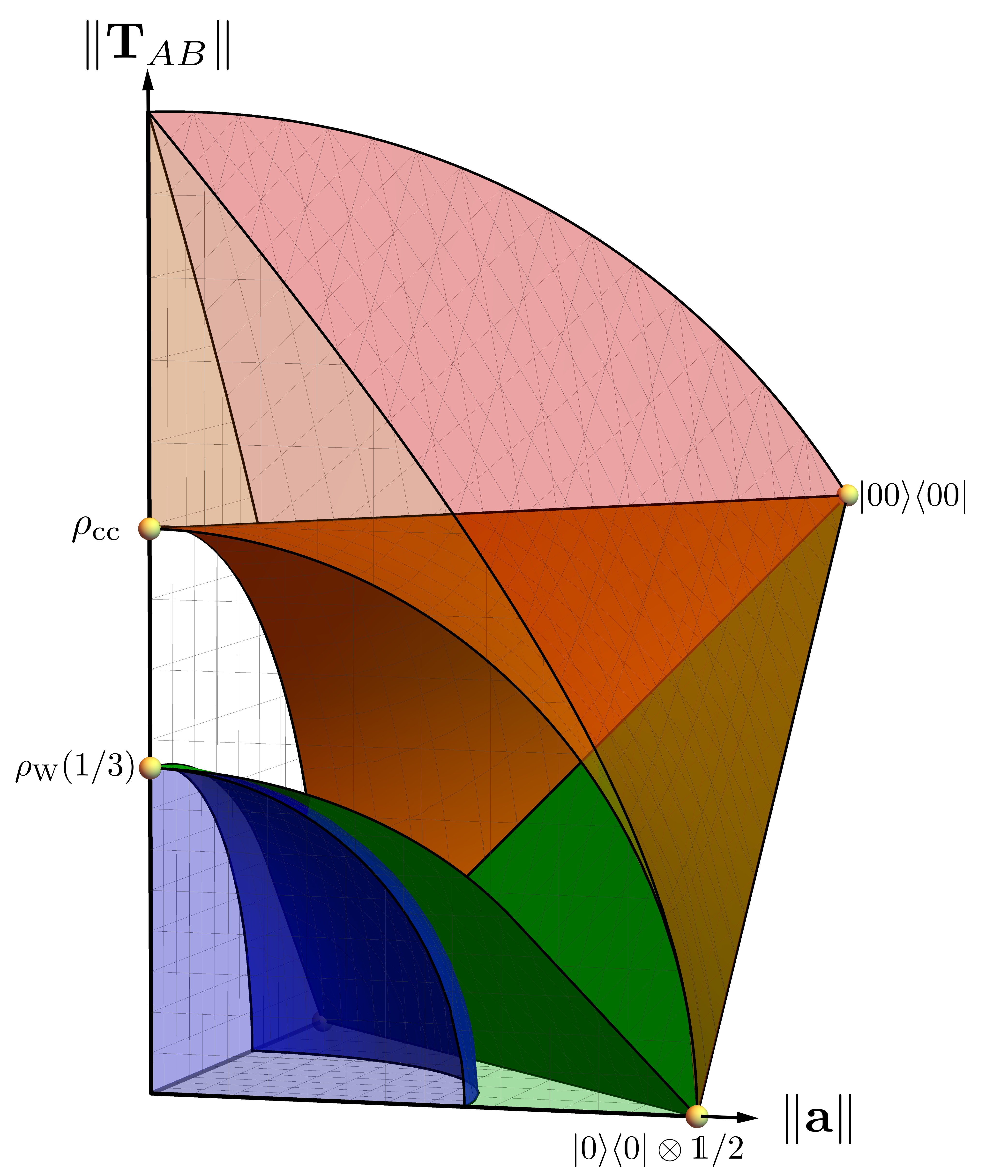}
  \caption{Entangled and separable region in our model. The purely entangled region defined by Eq.~(\ref{eq:cond_entanglement}) is shown in orange, being mapped into this region is already a sufficient condition for entanglement. In blue there is the separable ball of radius $1/\sqrt{3}$, that is extended to the purely separable region defined by Eq.~(\ref{eq:cond_separability}) shown in green.
  For every point in the area in between there exists both an entangled and a separable state mapping to that point. So these are in fact the strongest criteria for separability and entanglement that can be found in this model.}
  \label{fig:entanglement_bound}
\end{figure}

It was shown that the condition $S_L(\rho_{AB})<\max\{S_L(\rho_{A}),S_L(\rho_{B})\}$ on the linear entropy $S_L(\rho)=1-\Tr(\rho^2)$ implies that the state is entangled~\cite{Horodecki_1996}. Using $S_L(\rho_{A/B})=1-1/2(1+\|\mathbf{a/b}\|^2)$ and $S_L(\rho_{AB})=1-1/4(1+\|\mathbf{a}\|^2+\|\mathbf{b}\|^2+\|\mathbf{T}_{AB}\|^2)$ this can be translated into Bloch lengths~\cite{Imai_2021}
\begin{align}\label{eq:cond_entanglement}
    \|\mathbf{T}_{AB}\|^2>1-\big|\|\mathbf{a}\|^2-\|\mathbf{b}\|^2\big|\ \ .
\end{align}
This is a sufficient condition for entanglement, but it is not necessary. For example, the Werner states~\cite{Werner1989} for two qubits,
$\rho_{\text{W}}(p)=p|\phi^+\rangle\langle\phi^+|+(1-p)\mathds{1}/4$ are entangled for every $p>1/3$, or $\|\mathbf{T}_{AB}\|>1/\sqrt{3}$. Since $\|\mathbf{a}\|=\|\mathbf{b}\|=0$, the inequality would only detect entanglement for $\|\mathbf{T}_{AB}\|>1$.
But nonetheless it is the largest possible entangled region as, e.g., the state $\rho_{\text{cc}}=1/2(|00\rangle\langle00|+|11\rangle\langle11|$ is separable with $\|\mathbf{a}\|=\|\mathbf{b}\|=0$ and $\|\mathbf{T}_{AB}\|=1$.
It was shown that separable states can achieve equality 
in Eq.~(\ref{eq:cond_entanglement}) for all values of $\|\mathbf{a}\|$ and $\|\mathbf{b}\|$~\cite{Imai_2021}.
In Fig.~\ref{fig:entanglement_bound} this region is shown in orange.
A necessary condition for the violation of a Bell-inequality of CHSH form is given by $\|\mathbf{T}_{AB}\|>1$~\cite{HORODECKI1995340}.

A sufficient condition for separability was found by the authors of Refs.~\cite{Zyczkowski_1998,Gurvits_2002} in the form of a separable ball around the maximally mixed state. A state is separable if
\begin{align}
    \|\mathbf{a}\|^2+\|\mathbf{b}\|^2+\|\mathbf{T}_{AB}\|^2\le\frac{1}{3}\ \ .
\end{align}
This relation defines a ball of separable states of radius $1/\sqrt{3}$ around the origin, shown in blue in Fig.~\ref{fig:entanglement_bound}.

But this is not the largest separable set, i.e., a set where only separable states are mapped to. In fact, the inequality 
\begin{align}\label{eq:cond_separability}
    \|\mathbf{T}_{AB}\|\le\begin{cases} \sqrt{\tfrac{1}{6}(2-3(\|\mathbf{a}\|+\|\mathbf{b}\|)^2)}\quad &\|\mathbf{a}\|+\|\mathbf{b}\|\le\tfrac{2}{3}\\
    1-\|\mathbf{a}\|-\|\mathbf{b}\| &\tfrac{2}{3}<\|\mathbf{a}\|+\|\mathbf{b}\|\le1\\
    \|\mathbf{a}\|+\|\mathbf{b}\|-1 &1<\|\mathbf{a}\|+\|\mathbf{b}\|.
    \end{cases}
\end{align}
is a sufficient condition for the separability of a state and the strongest sufficient condition formulated in terms of $\|\mathbf{a}\|$, $\|\mathbf{b}\|$ and $\|\mathbf{T}_{AB}\|$. A derivation of this inequality can be found in Appendix~\ref{app:proof_separable}.
This extended region is shown in green in Fig.~\ref{fig:entanglement_bound}.

\section{Maximally entangled mixed states}
When trying to maximize the entanglement of a state under certain restrictions two families of states are of particular interest.
The first family consists of states that maximize the entanglement for a fixed global purity, referred to as maximally entangled mixed states (MEMS)~\cite{Ishizaka_2000,Munro_2001,Wei_2003}. For a two-qubit system they are of the form
\begin{align}\label{eq:mems}
    \rho_{\scalebox{0.5}{MEMS}}(x,\theta)=&x(|00\rangle\langle00|+|11\rangle\langle11|)+(1-2x-\theta)|01\rangle\langle01|\notag\\
    &+\theta/2(|00\rangle+|11\rangle)(\langle00|+\langle11|).
\end{align}
The entanglement, quantified by the Wootters concurrence~\cite{Wootters}, becomes $\mathcal{C}(\rho_{\scalebox{0.5}{MEMS}})=\theta$ for these states. For $0\le\mathcal{C}(\rho_{\scalebox{0.5}{MEMS}})\le2/3$ it is maximized by $x=1/3-\theta/2$ and in the model the states have coordinates $(1/3,1/3,\sqrt{2\theta^2+1/9})$.
For $2/3\le\mathcal{C}(\rho_{\scalebox{0.5}{MEMS}})\le1$ the states that maximize the concurrence are those with $x=0$ and have coordinates $(1-\theta,1-\theta,\sqrt{1-4\theta+6\theta^2})$.

\begin{figure}[t]
  \centering
  \includegraphics[width=.9\linewidth]{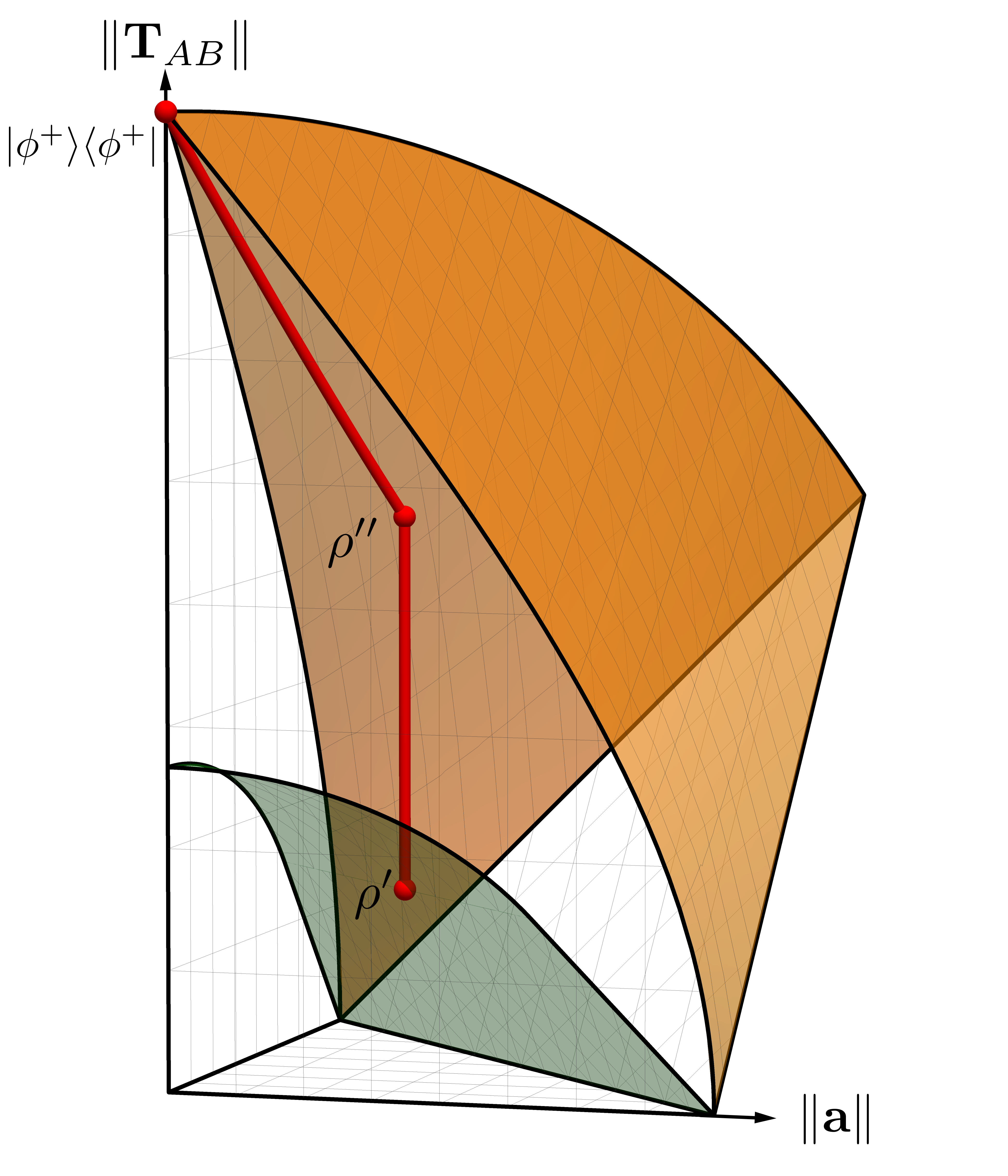}
  \caption{Maximally entangled states with respect to some purity constraint. The states maximizing the concurrence for a given purity (MEMS) are shown in red, they start at the separable region with the state $\rho'$ and continue in a straight line to the state $\rho''$. Here they bend and continue (not in a straight line) until the maximally entangled state. The states maximizing the concurrence for a given mixedness of the marginals (MEMMS) correspond to the states on the boundary of our model and are shown in orange.
    }
  \label{fig:mems_memms}
\end{figure}

In  Fig.~\ref{fig:mems_memms} they are depicted in red, where three special cases are highlighted: the state $\rho'=\rho_{\scalebox{0.5}{MEMS}}(1/3,0)$ is the only  unentangled state of this family located at the boundary of the separable region. The state $\rho''=\rho_{\scalebox{0.5}{MEMS}}(0,2/3)$ connects the two different pieces and incidentally also sits on the boundary of the entangled region. Finally the state $\rho_{\scalebox{0.5}{MEMS}}(0,1)=|\phi^+\rangle\langle\phi^+|$ maximizes the overall entanglement.

The other family of interest are maximally entangled states with fixed marginal mixedness (MEMMS)~\cite{Adesso_2003,Baio_2019}. That is, they maximize the two-qubit concurrence~\cite{Wootters} for given local Bloch lengths $\|\mathbf{a}\|$ and $\|\mathbf{b}\|$.
The states in Eq.~\eqref{eq:boundary} representing the upper boundary of our model of the state space are exactly the MEMMS for two qubits. They are shown in orange in Fig.~\ref{fig:mems_memms}.


\section{Conclusion}
We have introduced a novel inequality that relates the length of the local with global Bloch vector that is provably tight for qubits and provides deeper insight for qudits. 
By visualizing the three relevant quantities we provide an instructive representation of physically relevant properties in higher-dimensional state-spaces. Beyond aesthetic and didactic value, we believe that such simple relations will provide an important tool in the study of complex many-body systems, where complete information about the quantum state is unattainable for reasons of both experimental and theoretical complexity. 
One may expect that a deeper understanding of relations between sector 
lengths in multipartite systems may help to derive physically relevant
correlation constraints and thus make the action of the quantum marginal
problem more transparent. The solution of the simplest case presented in this
work gives a first indication of the path toward that objective. While our inequality is not necessarily tight in higher dimensions it remains a valuable tool, and strengthening the inequality is an open problem that will further broaden the application of this research program.


\section{Acknowledgements}
S.M. was supported by the Basque Government through IKUR strategy and through the BERC 2022-2025 program and by the Ministry of Science and Innovation: BCAM Severo Ochoa accreditation CEX2021-001142-S / MICIN / AEI / 10.13039/501100011033, PID2020-112948GB-I00 funded by MCIN/AEI/10.13039/501100011033 and by "ERDF A way of making Europe".
 M.H. acknowledges funding from the European Research Council (Consolidator grant ``Cocoquest'' 101043705) and the European Comission through the grant 'HyperSpace' (101070168).
 J.S. was supported by
 Grant PID2021-126273NB-I00 funded by MCIN/AEI/10.13039/501100011033 and by "ERDF A way of making Europe" as well as by
the Basque Government through Grant No.\ IT1470-22.

\bibliographystyle{apsrev4-1fixed_with_article_titles_full_names}
\bibliography{bibfile.bib}

\begin{thebibliography}{57}
\bibitem [{foo()}]{footnote1}%
  \BibitemOpen
  \href@noop {} {}\bibinfo {note} {We use the term correlation to refer to any
  quantity of the global system not visible locally, as opposed to the clearly
  defined terms correlation function and correlation tensor. The former refers
  to the normalized covariance between two observables, the latter to the
  collection of joint expectation values of a Bloch basis.}\BibitemShut {Stop}%
\end{thebibliography}


\appendix

\section{Separability from Bloch lengths}\label{app:proof_separable}

Here we prove that
\begin{align}
    \|\mathbf{T}_{AB}\|\le\begin{cases} \sqrt{\tfrac{1}{6}(2-3(\|\mathbf{a}\|+\|\mathbf{b}\|)^2)}\quad &\|\mathbf{a}\|+\|\mathbf{b}\|\le\tfrac{2}{3}\\
    1-\|\mathbf{a}\|-\|\mathbf{b}\| &\tfrac{2}{3}<\|\mathbf{a}\|+\|\mathbf{b}\|\le1\\
    \|\mathbf{a}\|+\|\mathbf{b}\|-1 &1<\|\mathbf{a}\|+\|\mathbf{b}\|
    \end{cases}\notag
\end{align}

is a sufficient criteria for separability.

Note that through local unitary rotations the marginals of every state can be diagonalized. That is, after a local basis change every state can be expanded in a Pauli product basis as

\begin{widetext}
\begin{align}
\rho_{AB} &= \frac{1}{4} \big(\mathds{1}_4 + a_{Z} \sigma_Z\otimes\mathds{1}_2 + b_{Z} \mathds{1}_2\otimes\sigma_Z + \sum_{X,Y,Z} c_{ij} \sigma_i\otimes\sigma_j \big)\\
    =&\frac{1}{4}\scalebox{1}{$\begin{pmatrix}1+a_Z+b_Z+c_{ZZ} & c_{XZ}-ic_{YZ} &  c_{ZX}-ic_{ZY} &  c_{XX}-c_{YY}-ic_{XY}-ic_{YX}\\
    c_{XZ}+ic_{YZ} & 1+a_Z-b_Z-c_{ZZ} & c_{XX}+c_{YY}-ic_{XY}+ic_{YX} & -c_{ZX}+ic_{ZY}\\
    c_{ZX}+ic_{ZY} & c_{XX}+c_{YY}+ic_{XY}-ic_{YX} & 1-a_Z+b_Z-c_{ZZ} & -c_{XZ}+ic_{YZ}\\
    c_{XX}-c_{YY}+ic_{XY}+ic_{YX} & -c_{ZX}-ic_{ZY} & -c_{XZ}-ic_{YZ} & 1-a_Z-b_Z+c_{ZZ}\end{pmatrix}$},\notag
\end{align}
\end{widetext}

where $a_Z,b_Z\ge0$.

Consider first the case $a_Z+b_Z>1$.
We know that $c_{ZZ}\ge a_Z+b_Z-1=\|\mathbf{a}\|+\|\mathbf{b}\|-1$.
Therefore states satisfying $\|\mathbf{T}_{AB}\|=\|\mathbf{a}\|+\|\mathbf{b}\|-1$ are diagonal and thus separable.
Now, recall that negativity under partial transpose is a necessary and sufficient condition for two qubits to be entangled and let $\rho^{\Gamma}$ be the partial transpose of $\rho$.
Note that $|\rho_{12}|=|\rho_{34}|=|\rho_{21}|=|\rho_{43}|$ and $|\rho_{13}|=|\rho_{24}|=|\rho_{31}|=|\rho_{42}|$ and all remain unchanged under the partial transposition.
Hence $\rho_{11}\rho_{44}\ge|\rho_{14}|^2$ and $\rho_{22}\rho_{33}\ge|\rho_{23}|^2$ are necessary and sufficient for positivity after partial transpose, equivalently
\begin{align*}
    (1+c_{ZZ})^2-(a_Z+b_Z)^2&\ge (c_{XX}+c_{YY})^2+(c_{XY}-c_{YX})^2\\
    (1-c_{ZZ})^2-(a_Z-b_Z)^2&\ge (c_{XX}-c_{YY})^2+(c_{XY}+c_{YX})^2\ .
\end{align*}

Hence the stronger conditions
\begin{align*}
    (1+c_{ZZ})^2-(a_Z+b_Z)^2&\ge 2(c_{XX}^2+c_{YY}^2+c_{XY}^2+c_{YX}^2)\\
    (1-c_{ZZ})^2-(a_Z-b_Z)^2&\ge 2(c_{XX}^2+c_{YY}^2+c_{XY}^2+c_{YX}^2)\ ,
\end{align*}
are sufficient for separability and finally also
\begin{align}
    \|\mathbf{T}_{AB}\|^2\le \tfrac{1}{2}\pm c_{ZZ} +\tfrac{3}{2}c_{ZZ}^2-\tfrac{1}{2}(a_Z\pm b_Z)^2
\end{align}
for both signs.
Choosing $c_{ZZ}$ negative we can just take the positive sign. 
Minimizing over $c_{ZZ}$ gives $c_{ZZ}=-1/3$ for $a_Z+b_Z\le2/3$ and $c_{ZZ}=a_Z+b_Z-1$ for $a_Z+b_Z>2/3$, resulting in $\|\mathbf{T}_{AB}\|\le\sqrt{\tfrac{1}{6}(2-3(\|\mathbf{a}\|+\|\mathbf{b}\|)^2)}$ and $\|\mathbf{T}_{AB}\|\le1-\|\mathbf{a}\|-\|\mathbf{b}\|$ respectively.

To see that this is indeed the optimal bound, choose $c_{ZZ}=a_Z+b_Z-1$ and $c_{XX}=c_{YY}=\epsilon/2$. This state is entangled for every $\epsilon\ge0$.
For the case $a_Z+b_Z\le2/3$ take $c_{ZZ}=1/3$ and $c_{XX}=c_{YY}=1/9 -(a_Z+b_Z)^2/4+\epsilon\le1/9 -(a_Z-b_Z)^2/4$, which is an entangled state for all $\epsilon>0$.

\end{document}